\documentclass[reqno,11pt]{article}
\usepackage{amsfonts}
\usepackage{amsmath}
\usepackage{enumerate}
\usepackage[parfill]{parskip}
\usepackage{color}

\usepackage{amsmath, amssymb}
\usepackage{epsfig,graphicx}
\usepackage{dsfont}
\usepackage{hyperref}
\usepackage[round]{natbib}

\definecolor{Red}{rgb}{0.00, 0.00, 0.00} 
    \newcommand{\Red}{\color{Red}}
\definecolor{Blue}{rgb}{0.00, 0.00, 0.00} 
    
\definecolor{Green}{rgb}{0.00, 0.00, 0.00}  

\newcommand{\nv}{\boldsymbol}
\newcommand{\1}[1]{{\boldsymbol 1_{\{#1\}}}}

\newcommand{\tbt}[4]{{\left(\!\!\begin{array}{cc} {#1} & {#2} \\ {#3} & {#4} \end{array}\!\!\right)}}

\newtheorem{theorem}{Theorem}

\newtheorem{corollary}{Corollary}

\newtheorem{definition}{Definition}

\newtheorem{lemma}{Lemma}

\newtheorem{proposition}{Proposition}

\newenvironment{proof}[1][Proof]{\textbf{#1.} }{\ \rule{0.5em}{0.5em}}

\newcommand{\F}{\mathcal{F}}
\newcommand{\Px}{ \mathbb{P} }
\newcommand{\Qx}{ \mathbb{Q} }
\newcommand{\Ex}{ \mathbb{E} }
\newcommand{\Gx}{\mathbb{G}}
\newcommand{\Fx}{\mathbb{F} }

\usepackage{setspace}
\begin{document}

\title{Default and Systemic Risk in Equilibrium}
\author{Agostino Capponi\thanks{School of Industrial Engineering, Purdue University, West Lafayette, IN, 47906, Email:capponi@purdue.edu} \and
Martin Larsson\thanks{School of Operations Research, Cornell University, Ithaca, NY, 14853, Email: mol23@cornell.edu}}
\date{\today}
\maketitle

\begin{abstract}
We develop a finite horizon continuous time market model, where risk averse investors maximize utility from terminal wealth by dynamically investing in a risk-free money market account, a stock written on a default-free dividend process, and a defaultable bond, whose prices are determined via equilibrium. We analyze the endogenous interaction arising  between the stock and the defaultable bond via the
interplay between equilibrium behavior of investors, risk preferences and cyclicality properties of the default intensity. We find that the equilibrium price of the stock experiences a jump at default, despite that the default event has no causal impact on the dividend process. 
We characterize the direction of the jump in terms of a relation between investor preferences and the cyclicality properties of the default intensity. We conduct a similar analysis for the market price of risk and for the investor wealth process, and determine how heterogeneity of preferences affects the exposure to default carried by different investors.
\end{abstract}
\quad
\section{Introduction}

The default of a systemically important entity can have an impact on the rest of the economy through a number of different mechanisms. For instance, firms that have exposures to the defaulted entity through market transactions, can experience a deterioration in fundamentals driving the value of their assets. Under adverse circumstances this can lead to a domino effect, where the default of one firm causes financial distress on entities with which the firm had business relations. This distress can propagate through the financial system causing a cascading failure, leading in the worst case to the collapse of a significant portion of the system (the recent credit crisis being a clear example). In the context of interbank lending, \cite{giesecke-weber} propose a reduced form contagion model, while \cite{ContMinAmi:2010} and \cite{ContMinAmi:2011} use tools from random graph theory to analyze short term counterparty credit exposures. Dynamic contagion models are considered in~\cite{Rungga:2009}, and more recently in ~\cite{Cvitanic/Ma/Zhang:2010} and \cite{Giesecke/Spiliopoulos/Sowers:2011}.

Alternatively, there may be a purely informational effect, where the default of one firm triggers the market participants to update their perception of the state of the economy. For example, \cite{Dufresne/Goldsein:2011} show that the unexpected default of an individual firm can lead to a market-wide increase in credit spreads, and demonstrate via calibration that the risk premium due to contagion risk may be considerable.

A third possibility is that the sudden shock associated with the default event leads to a re-allocation of wealth as the economy returns to equilibrium. This may in turn cause rapid price changes due to linkages that stem from the equality between supply and demand. The aim of the present paper is to study this mechanism in a continuous time financial model, including default risk, where prices are determined endogenously in equilibrium.

While models of economic equilibrium have been studied for a long time, it is only recently that fully dynamic stochastic models of equilibrium have received significant attention. \cite{Dumas:1989} considers a dynamic equilibrium model with two investors, and characterizes the equilibrium behavior of the wealth allocation and risk-free rate, assuming that the stock returns are specified exogenously. \cite{Chaba:2011} considers a similar economy, but allows for the possibility of portfolio constraints, and analyzes  cyclicality properties of market price of risk and stock return volatilities. \cite{Bhamra/Uppal:2009} consider a continuous time economy populated by two power utility agents with heterogenous beliefs and preferences, and give closed form expressions for consumption policies, portfolio policies, and asset prices. The same model as in \cite{Bhamra/Uppal:2009} is considered by \cite{CviNapp:2011} and~\cite{CviMal:2010}, who extend the results by Bhamra and Uppal to the case of an arbitrary number of agents, including an asymptotic analysis for large time horizons. ~\cite{CviMal:2011} provide decompositions into myopic and non-myopic components for market price of risk, stock volatility, and hedging strategies. In the same economic model, \cite{Wang:1996} studies how investor preferences affect the term structure of interest rates.

The literature on dynamic equilibrium models, including the papers mentioned above, has been concerned primarily with models where equilibrium prices have continuous paths. This means that dramatic and sudden changes, such as crisis events or major defaults, are absent---and indeed these papers have focused on other economic phenomena. An exception is \cite{Hasler:2011}, which considers a Lucas economy with multiple defaults, where the default intensities are constant.


In the present paper we study a finite horizon continuous time model, where rational investors maximize utility from terminal wealth. Three securities are liquidly and dynamically traded: a money market represented by a locally risk-free security,
i.e. investors can borrow from or lend to each other without default, a stock representing shares of the aggregate endowment, and a defaultable bond which represents the corporate bond index (for example, the Dow Jones corporate bond index). 
We assume a constant recovery model, in which case the default of the bond index is interpreted as the default of one (or more) of the index bonds, which reduces the total payment of the index. The intensity of the defaultable bond may, but need not, depend on the dividend process.



As we demonstrate in the present paper, introducing a defaultable security in the economy leads to new insights regarding the behavior of securities prices, market price of risk, and wealth allocation. For instance, we find that the equilibrium price of the stock typically jumps when default occurs, despite the fact that the underlying dividend process is entirely unaffected by the default event. Moreover, the direction of the jump (up or down) depends in a non-trivial way on the interplay between investor preferences and the cyclicality properties of the default intensity. In particular, we show that upward jumps in the stock price are possible if, roughly speaking, the default intensity is sufficiently counter-cyclical. The precise statement is given in Theorem~\ref{T:sj}. We also show that a similar analysis, with similar conclusions, can be carried out for the wealth processes of individual investors, see Section~\ref{S:W}. In this connection, we investigate how heterogeneity of preferences affects the exposure to the default carried by the different investors.

Due to the possibility of default, there are two sources of risk in our model: diffusion risk and jump risk. Using techniques from the theory of filtration expansions, which has a long and successful history in credit risk modeling, we are able to guarantee market completeness, even in the presence of jumps, see also \cite{BielJean:2006} and \cite{BielJeanRu:2006} for a detailed analysis of market completeness and replication strategies in reduced form models of credit risk.
This allows us to identify a unique market price of risk process, corresponding to diffusion risk, and default risk premium process, corresponding to jump risk. It turns out that the two quantities are intimately linked, see Proposition~\ref{P:drpsen}. By means of a quite delicate mathematical analysis, these quantities are studied in the case of constant interest rate and default intensity.

The most natural interpretation of the phenomena we study is as a form of systemic risk, arising in an economy consisting of securities
carrying both market and default risk.  
While systemic risk effects generated from equilibrium models have been studied, for instance in~\cite{Allen:2000} and~\cite{Freixas:2010}, these papers use static discrete time models, exclusive of default, where the focus is on characterizing optimal risk sharing across banks with different credit profiles, or belonging to different geographical sectors. Differently from most research efforts, our model exhibits an endogenous interaction between the stock and the defaultable bond, which arises via the interplay between equilibrium behavior of the investors and their risk preferences.

The rest of the paper is organized as follows. Section \ref{sec:model} introduces the economic model. Section \ref{sec:mpr} analyzes
the market price of risk in equilibrium along with its behavior at the default event. Section \ref{S:Sp} characterizes the behavior of the
equilibrium stock price at default via a relation between cyclicality properties of short rate and default intensity, and investor preferences.
Section \ref{S:W} performs a similar analysis for the wealth process of a risk-averse agent, and, in the case of a power utility investor, provides monotonicity relations between the size of the jump and the level of risk aversion. Section \ref{sec:concl} concludes the paper.
The proofs of the necessary lemmas are deferred to the appendix.
\section{The Model}\label{sec:model}

\subsection{The Probabilistic Model}

Let $(\Omega,\mathcal F,\Px)$ be a complete probability space, supporting a standard Brownian motion $B=(B_t)_{0\leq t\leq T}$. Let $\Fx=(\F_{t})_{0\leq t\leq T}$ be the augmented filtration generated by $B$, which satisfies the usual hypotheses of completeness and right continuity. We use a standard construction (also called Cox construction) of the default time $\tau$, based on doubly stochastic point processes, using a given nonnegative $\mathbb F$~adapted intensity process $\lambda=(\lambda_t)_{0\leq t\leq T}$. To this end, we assume the existence of an exponentially distributed random variable $\chi$ defined on the probability space $(\Omega,\F, \Px)$, independent of the process $B$.
The default time $\tau$ is then defined as
$$
\tau = \inf\{ t\geq 0: \int_0^t \lambda_s ds \geq \chi\}.
$$
The market filtration $\mathbb G = (\mathcal G_t)_{0\leq t\leq T}$, which describes the information available to investors, is given by
$$
\mathcal G_t  = \bigcap_{u>t} \mathcal F_u \vee \sigma(\tau\wedge u).
$$
That is, it contains all information in $\mathcal F_t$, together with the knowledge of whether $\tau$ has occurred or not, and has been made right-continuous. It is a well-known result (see e.g.~\cite{Bielecki/Rutkowski:2001}, Section~6.5 for details) that the process
$$
M_t = \1{\tau\leq t} - \int_0^{t\wedge\tau}\lambda_sds
$$
is a ${\Gx}$-martingale under $\Px$. In other words, $\lambda$ is the \emph{default intensity} (or hazard rate) of~$\tau$.

An {important consequence of the previous construction is that Hypothesis~(H) holds, i.e.~every $\mathbb F$~martingale remains a $\mathbb G$~martingale, see \cite{Bielecki/Rutkowski:2001}. It then follows from a result by Kusuoka (Theorem~\ref{T:mgr} in Appendix~\ref{AP:fe}) that every square integrable $\mathbb G$~martingale may be represented as a stochastic integral with respect to $M$ and~$B$.

\subsection{The market model}
We consider a market model, which is an extension of the standard setting in~\cite{Cvitanic/Malamud:2010}. We assume that there is an underlying dividend process $D=(D_t)_{0\leq t\leq T}$ with dynamics
\begin{equation}
\frac{dD_t}{D_t} = \mu^D(D_t) dt + \sigma^{D}(D_t) dB_t, \qquad D_0>0.
\label{eq:divid}
\end{equation}
It is assumed that $\mu^D:\mathbb R_+ \to\mathbb R$ and $\sigma^D:\mathbb R_+\to \mathbb R_+$ are such that a strictly positive, strong solution exists. We also assume that $\mu^D$ and $\sigma^D$ are infinitely differentiable on~$(0,\infty)$, and that $\sigma^D>0$.

There are two risky assets in the economy, a stock which carries market risk, and a defaultable bond which carries default risk. At terminal time $T$, the stock pays a terminal dividend $D_T$, while the defaultable bond pays a terminal dividend $P_T$. The latter is given by
$$
P_T = \1{\tau>T} + \varepsilon\1{\tau\leq T}.
$$
Here $0<\varepsilon<1$ is a constant recovery value paid at time $T$ in case default happens at or before $T$. We assume that $\varepsilon$ is deterministic, although many calculations would still be valid as long as $\varepsilon$ is $\mathcal F_T$-measurable. Neither the stock, nor the defaultable bond generates any intermediate dividends. We also assume the existence of a locally risk free money-market account with interest rate $r=(r_t)_{0\leq t\leq T}$. Finally, we assume that the default intensity $\lambda_t$ and interest rate $r_t$ are of the form
$$
\lambda_t = \lambda(D_t) \qquad\text{and}\qquad r_t = r(D_t)
$$
for deterministic functions $\lambda$ and $r$. The same assumption has also been used by~\cite{Cvitanic/Malamud:2010} for the interest rate.

In our model, both the stock and the defaultable bond 
are positive net supply assets. In contrast, the zero money-market account is assumed to be available in zero net supply.

The market price at time $t$ of the stock is denoted by $S_t$, and that of the defaultable bond by $P_t$. These processes are determined in equilibrium, and their dynamics is of the form
\begin{align*}
\frac{dS_t}{S_{t-}} &= \mu^{S}_t dt + \sigma^{S}_t dB_t + \rho^{S}_t dM_t \\
\frac{dP_t}{P_{t-}} &= \mu^{P}_t dt + \sigma^{P}_t dB_t + \rho^{P}_t dM_t.
\label{eq:priceprocess}
\end{align*}
The existence of such representations follows from Theorem~\ref{T:mgr} together with the fact that in equilibrium, both $S_t$ and $P_t$ are semimartingales with absolutely continuous finite variation parts. Furthermore, we conjecture that the matrix
$$
\tbt{\sigma^S_t}{\rho^S_t}{\sigma^P_t}{\rho^P_t}
$$
will be invertible in equilibrium. This immediately implies that the market is complete, via application of Theorem~\ref{T:mgr}. It is then well known, see e.g.~\cite{Cvitanic/Zapatero:2004}, that there exists a unique state-price density process
$$
\xi=(\xi_t)_{0\leq t\leq T}.
$$
The time $t$ price of a payoff $X$ received at time~$T$ is given by $\frac{1}{\xi_t}\Ex[\xi_T X\mid\mathcal G_t]$.

\subsection{The investors}

There are a finite number of investors, indexed by $k$, who optimize expected utility from final consumption. They are all assumed to have identical beliefs given by the historical probability $\Px$, but can have different utility functions $U_k$. These are assumed to be twice continuously differentiable, strictly concave, and satisfy Inada conditions at zero and infinity:
$$
\lim_{x\downarrow 0} U'_k(x) = \infty \qquad \text{and}\qquad \lim_{x\to\infty} U'_k(x) = 0.
$$
Two important measures of risk aversion, which will be used extensively in this paper, are the coefficients of absolute and relative risk aversion, both defined in~\cite{Pratt:1964}. The coefficient of absolute risk aversion is defined as
\begin{equation}
\ell_U(x) = -\frac{d \log U'(x)}{dx} = - \frac{U''(x)}{U'(x)}.
\label{eq:absriskav}
\end{equation}
Pratt related this measure to the agent's risk behavior by showing that an agent with utility $U(x)$ is more risk averse than an agent with utility $V(x)$ if and only if {\Red{$\ell_U(x) > \ell_V(x)$}} for all $x \geq 0$. The coefficient of relative risk aversion is defined as
\begin{equation}
L_U(x) = -\frac{d \log U'(x)}{d \log x} = - \frac{xU''(x)}{U'(x)}.
\label{eq:relriskav}
\end{equation}

The $k$:th investor chooses a dynamic portfolio strategy $\pi_k = (\pi^S_{kt},\pi^P_{kt})_{0\leq t\leq T}$, a $\mathbb G$~predictable and $(S,P)$-integrable process, where $\pi^S_{kt}$ is the proportion of wealth invested in the stock at time $t$, and $\pi^P_{kt}$ is the proportion of wealth invested in the defaultable bond. The remaining wealth is invested in the money market account to make the strategy self-financing. The investor must choose his strategy so that the corresponding wealth process, given by
\begin{equation} \label{eq:Wkt}
\frac{dW_{kt}}{W_{kt-}} = r_t dt + \pi^S_{kt} \Big( \frac{dS_t}{S_{t-}} - r_t dt \Big) + \pi^P_{kt} \Big( \frac{dP_t}{P_{t-}} - r_t dt \Big),
\end{equation}
stays strictly positive for $0\leq t\leq T$. The portfolio strategy $\pi_k$ is chosen to maximize the expected utility
$$
\Ex\left[ U_k(W_{kT} )\right].
$$
Market completeness allows one to use standard duality methods (see~\cite{Cvitanic/Malamud:2010}) to show that the optimal final wealth in equilibrium is given by
\begin{equation} \label{eq:Wkt2}
W_{kT} = I_k(y_k\xi_T),
\end{equation}
where the number $y_k$ is the solution to the budget constraint equation,
$$
\Ex\left[ I_k(y_k \xi_T) \xi_T \right] = W_{k0}.
$$
Moreover, the wealth at times $t<T$ is given by
\begin{equation}\label{eq:Wkexp}
W_{kt} = \frac{\Ex[ \xi_T W_{kT} \mid \mathcal G_t ]}{\xi_t}.
\end{equation}

\subsection{The equilibrium}

We employ the usual notion of equilibrium:

\begin{definition}
The market is said to be in equilibrium if each investor behaves optimally and all the securities markets clear.
\end{definition}

Again by market completeness, standard equilibrium theory, see ~\cite{Const:1982}, shows that security prices coincide with those in an artificial economy populated by a single, representative investor. We denote the corresponding utility function by $U$, and assume that $U$ is twice continuously differentiable, strictly concave, and satisfies Inada conditions at zero and infinity. The state-price density is then given by
\begin{equation}
\xi_T = U'(D_T + P_T).
\label{eq:statepricedens}
\end{equation}
Furthermore,
$$
\xi_t = e^{-\int_0^t r_s ds} Z_t,
$$
where $Z$ is the Radon-Nikodym density process corresponding to the (unique) risk-neutral measure $\Qx$,
$$
Z_t = \frac{d\Qx}{d\Px}\Big|_{\mathcal G_t} = \Ex\Big[ e^{\int_0^T r_sds} \xi_T \mid\mathcal G_t \Big].
$$

Using Equation~(\ref{eq:statepricedens}), the definition of $D_T$ and $P_T$, and Lemma~\ref{L:ce} in Appendix~\ref{AP:fe}, we can separate the state price density into a pre- and post-default component. More precisely, we have
$$
\xi_t = \1{\tau > t} \xi_t^{\textit{pre}} +  \1{\tau \leq t} \xi_t^{\textit{post}},
$$
where
\begin{align}
\nonumber \xi_t^{\textit{post}} &= \Ex \bigg[e^{\int_t^T r_s ds} U'(D_T + \varepsilon) \bigg| \mathcal{F}_t  \bigg]  \\
\nonumber \xi_t^{\textit{pre}} &= \Ex\bigg[\left(1-e^{-\int_t^T \lambda_s ds} \right) e^{\int_t^T r_s ds} U'(D_T + \varepsilon)  \\
& \qquad\qquad\qquad + e^{-\int_t^T \lambda_s ds} e^{\int_t^T r_s ds} U'(D_T + 1) \bigg| \mathcal{F}_t \bigg].
\label{eq:xiprepost}
\end{align}


%

{\bf Remark.} Assume that the intensity $\lambda_t$ is deterministic and, for simplicity, that $r_t\equiv 0$. We then have
$$
\xi_t^{\textit{pre}} = \Px(\tau \leq T | \tau > t)\Ex[ U'(D_T + \varepsilon)\mid\mathcal F_t] + \Px(\tau > T | \tau > t) \Ex[ U'(D_T + 1)\mid\mathcal F_t],
$$
indicating that the pre-default state price density is the weighted average of the state price density in an economy where default will surely happen, and the state price density in a default-free economy. The weights are, respectively, the probability that default will, or will not, take place before $T$, given that it has not occurred up to time $t$.

The equilibrium market price processes are computed using the state price density $\xi_t$. They are given by
\begin{align}
S_t &= \frac{\Ex[ \xi_T D_T \mid \mathcal G_t ]}{\xi_t} = \Ex^\Qx\Big[ e^{-\int_t^T r_udu} D_T \mid\mathcal G_t \Big] \label{eq:marketpricesec} \\
P_t &= \frac{\Ex[ \xi_T P_T \mid \mathcal G_t ]}{\xi_t} = \Ex^\Qx\Big[ e^{-\int_t^T r_udu} P_T \mid\mathcal G_t \Big]. \nonumber
\end{align}
Therefore, again relying on Lemma~\ref{L:ce} in Appendix~\ref{AP:fe}, we obtain
\begin{equation}\label{eq:Sprepost1}
S_t = \1{\tau > t} S_t^{\textit{pre}} +  \1{\tau \leq t} S_t^{\textit{post}},
\end{equation}
where
\begin{eqnarray}
\nonumber S_t^{\textit{post}} &=& \frac{1}{\xi^{\textit{post}}_t} \Ex\left[D_T U'(D_T + \epsilon) \bigg| \mathcal{F}_t \right] \\
\nonumber S_t^{\textit{pre}} &=& \frac{1}{\xi^{\textit{pre}}_t} \Ex\left[\left(1-e^{-\int_t^T \lambda_s ds} \right) D_T U'(D_T + \epsilon) +
e^{-\int_t^T \lambda_s ds} D_T U'(D_T + 1)  \bigg| \mathcal{F}_t  \right]. \\
\label{eq:Sprepost2}
\end{eqnarray}

\section{Equilibrium market price of risk} \label{sec:mpr}

In this section we derive expressions for the market price of (diffusion and default) risk, as well as the risk premium of the stock. The risk premium is defined as the excess growth rate of the asset above the risk-free rate, namely $\mu^S_t-r_t$.

By Theorem~\ref{T:mgr} the density process $Z$ associated with the risk-neutral measure has the representation
$$
\frac{dZ_t}{Z_{t-}} = -\theta_t dB_t + \kappa_t dM_t
$$

for some $\mathbb G$~predictable processes $\theta$ and $\kappa$. An application of Girsanov's theorem shows that
$$
B^\Qx_t = B_t + \int_0^t \theta_s ds \qquad\text{and}\qquad M^\Qx_t = M_t - \int_0^{t\wedge\tau} \kappa_s\lambda_sds
$$
are $(\mathbb G,\Qx)$~local martingales, and in particular $B^\Qx$ is $(\mathbb G,\Qx)$~Brownian motion. Note that we can write
$$
M^\Qx_t = \1{\tau\leq t} - \int_0^{t\wedge\tau} \lambda_s(1+\kappa_s)ds,
$$
so that the risk-neutral default intensity is given by $\lambda^\Qx_t = \lambda_t (1+\kappa_t)$. The quantity $\kappa_t$ is called the $\emph{default risk premium}$, and $\theta_t$ is called the \emph{market price of risk}. We fix this notation from now on.

\begin{proposition} \label{P:mpretc}
The market price of risk is given by
$$
\theta_t = \theta^{\textit{pre}}_t\1{\tau\geq t}+\theta^{\textit{post}}_t\1{\tau< t},
$$
where $-\theta^{\textit{pre}}$ is the volatility of $\xi^{\textit{pre}}$, and $-\theta^{\textit{post}}$ is the volatility of $\xi^{\textit{post}}$. The default risk premium is given by
$$
\kappa_t = \frac{\xi^{\textit{post}}_t}{\xi^{\textit{pre}}_t} - 1.
$$
The risk premium associated with the stock, or the \emph{equity risk premium}, is given by
$$
\mu^S_t - r_t = \sigma^S_t\theta_t - \left( \frac{S^{\textit{post}}_t}{S^{\textit{pre}}_t} - 1 \right)\left( \frac{\xi^{\textit{post}}_t}{\xi^{\textit{pre}}_t} - 1 \right) \lambda_t \1{\tau\geq t}.
$$
\end{proposition}

\begin{proof}
The assertions concerning $\theta$ and $\kappa$ follow from Lemma~\ref{L:jp} and the definition of $\theta_t$ and $\kappa_t$, since $\xi_t = e^{-\int_0^t r_s ds}Z_t$. Let us establish the expression for the risk premium. The relations between $B$ and $B^\Qx$, respectively $M$ and $M^\Qx$, together with the $\Px$-dynamics of the stock price yield
$$
\frac{dS_t}{S_{t-}} = \left[ \mu^S_t - \sigma^S_t\theta_t + \rho^S_t \kappa_t\lambda_t\1{\tau\geq t} \right] dt + \sigma^S_t dB^\Qx_t + \rho^S_t dM^\Qx_t.
$$
The drift term equals $r_t dt$ since the discounted stock price is a martingale under $\Qx$. The proof follows by substituting the expressions for $\kappa_t$ and $\rho^S_t$ into the above equation (the latter follows from Lemma~\ref{L:jp}.)
\end{proof}

{\bf Remark.}
The risk premium can alternatively be expressed in terms of the risk-neutral default intensity $\lambda^\Qx_t$, using that $\lambda^\Qx_t = \lambda_t (1+\kappa_t)$. The result is
$$
\mu^S_t - r_t = \sigma_t\theta_t - \left( \frac{S^{\textit{post}}_t}{S^{\textit{pre}}_t} - 1 \right)\left( 1 - \frac{\xi^{\textit{pre}}_t}{\xi^{\textit{post}}_t} \right) \lambda^\Qx_t \1{\tau\geq t}.
$$

It is clear from the definition of $\xi^{\textit{pre}}_t$ and $\xi^{\textit{post}}_t$ that we always have $\xi^{\textit{pre}}_t \leq \xi^{\textit{post}}_t$.
{\Red The contribution to the equity risk premium coming from default risk therefore has the same sign as $S^{\textit{pre}}_t - S^{\textit{post}}_t$. This quantity is minus the size of the jump in the stock price, \emph{were default to happen at time $t$}. In particular, if the stock price jumps down at default, then the investors \emph{require} a premium for holding the stock, as they want to be compensated for the loss incurred upon default. On the other hand, if the stock jumps up at default, then it becomes an attractive security to hold, and therefore the investors are willing to \emph{pay} a premium for holding it. We will study the sign of the jump in more detail in Section~\ref{S:Sp}; suffice it to say here that positive price jumps, while atypical, are indeed possible.}

There is an interesting relationship between the sensitivity of $\kappa_t$ with respect to changes in the level of the dividend process, and the market price of diffusion risk. To state the result, first observe that the Markovian structure allows us to write
$$
\kappa_t = \kappa(t,D_t)
$$
for some measurable function $\kappa(t,x)$. We now have

\begin{proposition} \label{P:drpsen}
The function $\kappa$ is differentiable, and the derivative $\kappa_x=\frac{\partial \kappa}{\partial x}$ is given by
$$
\kappa_x(t,D_t) = - \frac{1}{D_t \sigma^D(D_t)} \frac{\xi^{\textit{post}}_t}{\xi^{\textit{pre}}_t} \left(\theta^{\textit{post}}_t -  \theta^{\textit{pre}}_t \right)
$$
\end{proposition}

\begin{proof}
As for $\kappa_t$, the Markovian structure allows us to write $\xi^i_t = \xi^i(t,D_t)$ for $i\in\{\textit{pre}, \textit{post}\}$ and measurable functions $\xi^i(t,x)$. As in the proof of Theorem~\ref{T:mprj} below, we may apply Theorem~6.1 in~\cite{Janson/Tysk:2006} to obtain the smoothness of $\xi^i$, and hence of $\kappa$ since $\kappa=\frac{\xi^{\textit{post}}}{\xi^{\textit{pre}}}-1$ by Proposition~\ref{P:mpretc}. Differentiating this relation yields
$$
\kappa_x = \frac{\xi^{\textit{post}}}{\xi^{\textit{pre}}} \left(\frac{\xi^{\textit{post}}_x}{\xi^{\textit{post}}} - \frac{\xi^{\textit{pre}}_x}{\xi^{\textit{pre}}} \right).
$$
Now, the volatility of a positive $\mathbb F$~adapted semimartingale of the form $u(t,D_t)$ is given by $\frac{u_x}{u}(t,D_t)D_t\sigma^D(D_t)$, as can be seen from It\^o's formula. By Proposition~\ref{P:mpretc},  $\theta^i$ is equal to minus the volatility of $\xi^i$, which yields the result.
\end{proof}

Observe that $\theta^{\textit{post}}_t -  \theta^{\textit{pre}}_t$ is the size of the jump in $\theta$, if default were to occur at time~$t$. Proposition~\ref{P:drpsen} shows in particular that if this quantity is positive, the default risk premium moves in the opposite direction to the dividend: an increase in the dividend process is accompanied by a decrease in the default risk premium, and vice versa. This appears to suggest that, upon default, a risk averse investor who sees an upward jump in the market price of risk, prefers to shift wealth from the risky stock to a default-free bond, giving a sure payoff of $\epsilon$ at maturity. If, on the other hand,  $\theta^{\textit{post}}_t -  \theta^{\textit{pre}}_t$ is negative, the default risk premium moves in the same direction as the dividend.

We proceed to study how the market price of risk $\theta_t$ behaves at default. As we have just seen, this also provides information about the sensitivity of the default risk premium $\kappa_t$ to changes in~$D_t$. The following result unfortunately requires us to assume constant interest rate and constant default intensity---already in this case the analysis is non-trivial (in particular it is much more delicate than for the jump in the stock price.) Extending it to more general $r$ and $\lambda$ is an interesting problem that we leave for future research.

\begin{theorem} \label{T:mprj}
Assume that the interest rate and default intensity are constant. If the representative investor has a strictly decreasing absolute risk aversion, then the market price of risk has a nonnegative jump at $\tau$.
\end{theorem}

The rest of this section is devoted to proving Theorem~\ref{T:mprj}. First, let us introduce some notation. For each $\alpha>0$, define the function
$$
u^\alpha(t,x) = \Ex[ U'(D_T + \alpha) \mid D_t=x].
$$
Using, for instance, Theorem~6.1 in \cite{Janson/Tysk:2006}, we deduce that $u^\alpha$ satisfies the PDE
$$
u^{\alpha}_t + \frac{1}{2}x^2 \sigma^D(x)^2 u^{\alpha}_{xx} + x \mu^D(x) u^{\alpha}_x = 0, \qquad u^{\alpha}(T,x) = U'(x+\alpha),
$$
where the subscripts denote partial derivatives. Standard results then imply that $u^\alpha$ has the same degree of smoothness as $\sigma^D$ and $\mu^D$ on $(0,T)\times (0,\infty)$, see e.g.~Theorem~10 in Chapter~3 of~\cite{Friedman:2008}. Since we assume that $\sigma^D$ and $\mu^D$ are infinitely differentiable, the same holds for $u^\alpha$.

\begin{proof}[Proof of Theorem~\ref{T:mprj}]
Due to Lemma~\ref{L:mprj2} below, the theorem will be proved once we establish that the quantity
$$
- \frac{\partial}{\partial x} \log u^\alpha(t,x)
$$
is decreasing in $\alpha$. This is done in two stages: Lemma~\ref{L:mprj3} gives the result when $D$ is bounded, and Lemma~\ref{L:mprj4} then extends this to unbounded $D$.
\end{proof}

\begin{lemma} \label{L:mprj2}
Assume that the interest rate and default intensity are constant. If
$$
-\frac{u^\varepsilon_x(t,x)}{u^\varepsilon(t,x)} > - \frac{u^1_x(t,x)}{u^1(t,x)}
$$
for all $(t,x)\in (0,T]\times \mathbb R_+$, then $\Delta \theta_\tau > 0$ on $\{\tau\leq T\}$.
\end{lemma}

\begin{proof}
It follows from~(\ref{eq:xiprepost}) and the assumption of constant $r$ and $\lambda$ that
$$
\xi^{\textit{post}}_t = e^{r(T-t)}u^\varepsilon(t,D_t)
$$
and
$$
\xi^{\textit{pre}}_t = e^{r(T-t)}\Big( (1-e^{-\lambda(T-t)}) u^\varepsilon(t,D_t) + e^{-\lambda(T-t)} u^1(t,D_t) \Big).
$$
The volatility of a positive $\mathbb F$~adapted semimartingale of the form $u(t,D_t)$ is given by $\frac{u_x}{u}(t,D_t)D_t\sigma^D(D_t)$, as can be seen from It\^o's formula. By Proposition~\ref{P:mpretc} and the above expressions for $\xi^{\textit{pre}}_t$ and $\xi^{\textit{post}}_t$ it then follows that
$$
\theta^{\textit{post}}_t = - \frac{u^\varepsilon_x(t,D_t)}{u^\varepsilon(t,D_t)}D_t\sigma^D(D_t)
$$
and
$$
\theta^{\textit{pre}}_t = - \frac{(1-e^{-\lambda(T-t)}) u^\varepsilon_x(t,D_t) + e^{-\lambda(T-t)} u^1_x(t,D_t)}{(1-e^{-\lambda(T-t)}) u^\varepsilon(t,D_t) + e^{-\lambda(T-t)} u^1(t,D_t)}D_t\sigma^D(D_t).
$$
A calculation using that $u^\varepsilon$ and $u^1$ are strictly positive reveals that $\theta^{\textit{post}}_t > \theta^{\textit{pre}}_t$ if and only if
$$
-\frac{u^\varepsilon_x(t,D_t)}{u^\varepsilon(t,D_t)} > -\frac{u^1_x(t,D_t)}{u^1(t,D_t)}.
$$
The result now follows.
\end{proof}

\begin{lemma} \label{L:mprj3}
Assume that the conditions of Theorem~\ref{T:mprj} are satisfied. Assume also that there is a constant $C>0$ such that $\sigma^D(x)=0$ and $\mu^D(x)=0$ for all $x\notin (C^{-1}, C)$. Then
$$
- \frac{\partial}{\partial x} \log u^\alpha(t,x)
$$
is strictly decreasing in $\alpha$.
\end{lemma}

\begin{proof}
Define $\widetilde u^\alpha = \log u^\alpha$. It can be readily verified that $\widetilde u^\alpha$ satisfies the terminal value problem
\begin{align*}
\widetilde u^\alpha_t + \frac{1}{2}x^2\sigma^D(x)^2 \widetilde u^\alpha_{xx} + x\mu^D(x) \widetilde u^\alpha_x + \frac{1}{2}x^2\sigma^D(x)^2 (\widetilde u^\alpha_x)^2 &= 0, \\
\widetilde u^\alpha(T,x) &= \log U'(x+\alpha).
\end{align*}
Now define $v^\alpha = - \widetilde u^\alpha_x =- \frac{\partial}{\partial x} \log u^\alpha$, and differentiate the above equation with respect to $x$ to see that $v^\alpha$ satisfies the nonlinear PDE
\begin{align*}
v^\alpha_t + \frac{1}{2}x^2\sigma^D(x)^2 v^\alpha_{xx} &+ \left( x\mu^D(x) + \frac{1}{2} [ x^2\sigma^D(x)^2]_x \right) v^\alpha_x  \\
&  + [x\mu^D(x)]_x v^\alpha - \frac{1}{2}[ x^2\sigma^D(x)^2 (v^\alpha)^2 ]_x = 0,
\end{align*}
with terminal condition
$$
v^\alpha(T,x) = -\frac{U''(x+\alpha)}{U'(x+\alpha)} = \ell_U(x+\alpha).
$$
Let us pick $\beta<\alpha$, and define $w = v^\beta - v^\alpha$. We want to prove that $w > 0$. The function $w$ satisfies the terminal value problem
\begin{align} \label{eq:mprj1}
w_t + \frac{1}{2} a(x) w_{xx} + b(t,x) w_x + c(t,x) w  &= 0 \\
\nonumber w (T,x) &=  \ell_U(x+\beta) -  \ell_U(x+\alpha),
\end{align}
where
\begin{align*}
a(x) &= x^2\sigma^D(x)^2 \\
b(t,x) &= x\mu^D(x) + \frac{1}{2} [ x^2\sigma^D(x)^2]_x - \frac{1}{2} x^2\sigma^D(x)^2 (v^\alpha(t,x)+v^\beta(t,x)) \\
c(t,x) &= \big[x\mu^D(x) - \frac{1}{2} x^2\sigma^D(x)^2 (v^\alpha(t,x)+v^\beta(t,x))\big]_x.
\end{align*}
Notice that $w (T,x) > 0$, as we are assuming that the coefficient of absolute risk aversion $\ell_U$ is strictly decreasing. Moreover, the coefficients $a, b$ and $c$ are smooth due to the smoothness of $\mu^D$, $\sigma^D$, $v^\alpha$ and $v^\beta$. The latter functions are smooth since they are the derivatives of the logarithm of the infinitely differentiable functions $u^\alpha$ and $u^\beta$.

Now, let $X=(X_t)_{0\leq t\leq T}$ be the solution to the SDE
$$
dX_t = \sqrt{a(X_t)} dB_t + b(t,X_t) dt, \qquad X_0=D_0.
$$
The smoothness of $a$ and $b$ implies that a unique strong solution exists up to an explosion time, but since $\sigma^D(x)=0$ and $\mu^D(x)=0$ for all $x\notin (C^{-1}, C)$, we have $a(x)=0$ and $b(t,x)=0$ there, so no explosion can occur. Indeed, $C^{-1}\leq X_t \leq C$ holds for $0\leq t\leq T$, almost surely.

Next, define a process $Y=(Y_t)_{0\leq t\leq T}$ by
$$
Y_t =  e^{\int_0^t c(s,X_s) ds}w(t,X_t).
$$
It\^o's formula and the fact that $w$ satisfies~(\ref{eq:mprj1}) show that
$$
dY_t =  e^{\int_0^t c(s,X_s) ds}w_x(t,X_t) \sqrt{a(X_t)} dB_t,
$$
and since $X_t$ remains in a compact set and $a$, $c$ and $w_x$ are continuous, the integrand in front of $dB_t$ is bounded. Therefore $Y$ is a martingale, and its final value is $Y_T=e^{\int_0^T c(s,X_s) ds}w(T,X_T) > 0$ due to the boundary condition of $w$. We deduce that $Y_t>0$ for every $t$ almost surely, and hence that $w > 0$, as desired.
\end{proof}

\begin{lemma} \label{L:mprj4}
Assume that the conditions of Theorem~\ref{T:mprj} are satisfied. Then
$$
- \frac{\partial}{\partial x} \log u^\alpha(t,x)
$$
is nonincreasing in $\alpha$.
\end{lemma}

\begin{proof}
Fix $\beta>\alpha$. The goal is to show that $-u^\alpha_x/u^\alpha \geq -u^\beta_x/u^\beta$. For each $n\in\mathbb N$, let $\mu^n$ and $\sigma^n$ be infinitely differentiable and coincide with $\mu^D$, respectively $\sigma^D$, on $[n^{-1},n]$, while being zero outside the interval $[(n+1)^{-1},n+1]$. Denote by $D^n$ the solution to
$$
\frac{dD^n_t}{D^n_t} = \mu^n(D^n_t) dt + \sigma^n(D^n_t) dB_t, \qquad D^n_0=D_0,
$$
and define $u^{\alpha,n}(t,x) = \Ex\left[ U'(D^n_T + \alpha) \mid D^n_t=x\right]$. An application of Lemma~\ref{L:mprj3} shows that
$$
-\frac{u^{\alpha,n}_x}{u^{\alpha,n}} > -\frac{u^{\beta,n}_x}{u^{\beta,n}}
$$
for each $n$. It thus suffices to prove that $u^{\alpha,n}\to u^\alpha$ and $u^{\alpha,n}_x\to u^\alpha_x$ pointwise. The latter follows from the former using \emph{interior Schauder estimates}, for instance by applying the corollary of Theorem~15 in Chapter~3 of~\cite{Friedman:2008} on each subdomain $[0,T)\times(m^{-1},m)$, $m\geq 2$ (using the PDE representation of $u^{\alpha,n}$, and noticing that on each subdomain the coefficients of the parabolic operator associated to $u^{\alpha,n}$ are H\"older continuous, and $x^2 \sigma^n(x)^2$ is bounded away from zero for all sufficiently large $n$.)

To prove that $u^{\alpha,n}(t,x) \to u^\alpha(t,x)$, first note that $u^{\alpha,n}(t,x) = \Ex\left[ U'(D^n_{T-t} + \alpha) \mid D^n_0=x\right]$ and $u^{\alpha}(t,x) = \Ex\left[ U'(D_{T-t} + \alpha) \mid D_0=x\right]$ by the Markov property. Since $U'(\cdot+\alpha)$ is bounded, the desired convergence follows from the Bounded Convergence Theorem if $D^n_{T-t} \to D_{T-t}$ almost surely, with $D^n_0=D_0=x$. But this is clear: pathwise uniqueness and the construction of $\mu^n$ and $\sigma^n$ imply that $D$ and $D^n$ coincide on the event
$$
A_n = \{n^{-1} \leq D_s \leq n \ \text{for\ all\ } 0\leq s\leq T-t \},
$$
so $D^n_{T-t} = D_{T-t} \nv 1_{A_n} + D^n_{T-t}\nv 1_{A_n^c}$. Since $\Px(A_n)\to 1$, $D^n_{T-t}\to D_{T-t}$ almost surely, and the proof is finished.
\end{proof}

\section{Equilibrium stock price} \label{S:Sp}

In this section we are interested in how the market price of the stock changes when default occurs. If $\tau< T$, there may be a jump in the stock price at~$\tau$. Under certain cyclicality assumptions on the default intensity and the interest rate, it turns out that the sign of the jump must be negative. On the other hand, in specific circumstances it can happen that the jump is positive. The following results gives the precise conditions. The proofs rely on a number of lemmas, which are stated and proved in Appendix~\ref{AP:stdom}.

\begin{theorem} \label{T:sj}
Assume that the interest rate is counter-cyclical, and that the representative investor has strictly decreasing absolute risk aversion, as well as relative risk aversion bounded by one. Define
\begin{equation}
g(t,x)=\Ex[e^{-\int_t^T\lambda_u du} \mid D_T=x] \qquad \text{and} \qquad \phi(x) = 1 - \frac{U'(x+1)}{U'(x+\varepsilon)}.
\label{eq:gdef}
\end{equation}
Then the following hold.
\begin{itemize}
\item[(i)] If $\phi(x) g(t,x)$ is strictly increasing in $x$ for every $0 \leq t\leq T$, the stock price has a strictly positive jump at $\tau$.
\item[(ii)] If $\phi(x) g(t,x)$ is strictly decreasing in $x$ for every $0\leq t\leq T$, the stock price has a strictly negative jump at $\tau$.
\end{itemize}
\end{theorem}

\begin{proof}
Equations (\ref{eq:xiprepost}), (\ref{eq:Sprepost1}) and (\ref{eq:Sprepost2}) show that the jump in the stock price is given by
\begin{equation} \label{eq:j0}
\Delta S_\tau = \frac{a_t}{c_t} - \frac{a_t-b_t}{c_t-d_t} \Bigg|_{t=\tau} \qquad \text{on}\quad \{0<\tau\leq T\},
\end{equation}
where
\begin{align*}
a_t &= \Ex\left[ D_T U'(D_T+\varepsilon) \mid \mathcal F_t \right] \\
b_t &= \Ex\left[ e^{-\int_t^T \lambda_u du} D_T ( U'(D_T+\varepsilon) - U'(D_T+1) )  \mid \mathcal F_t \right] \\
c_t &= \Ex\left[ e^{\int_t^T r_u du} U'(D_T+\varepsilon) \mid \mathcal F_t \right] \\
d_t &= \Ex\left[ e^{\int_t^T (r_u-\lambda_u) du} ( U'(D_T+\varepsilon) - U'(D_T+1) )  \mid \mathcal F_t \right).
\end{align*}
Using that $\xi^{\textit{pre}}_t = c_t-d_t$ and $S^{\textit{post}}_t = \frac{a_t}{c_t}$, elementary manipulations yields
\begin{align*}
\Delta S_\tau &= \frac{1}{\xi^{\textit{pre}}_t} \Bigg[ \text{Cov}_t\left(e^{-\int_t^T\lambda_u du} \phi(D_T), D_T U'(D_T+\varepsilon) \right)  \\
& \qquad\qquad - S^{\textit{post}}_t \text{Cov}_t\left(e^{-\int_t^T\lambda_u du} \phi(D_T), e^{\int_t^Tr_udu} U'(D_T+\varepsilon) \right) \Bigg]_{t=\tau}
\end{align*}
on $\{0<\tau\leq T\}$, where $\text{Cov}_t$ denotes $\mathcal F_t$-conditional covariance, and $\phi$ is defined in~(\ref{eq:gdef}). It suffices to analyze the two covariances, since both $\xi^{\textit{pre}}_t$ and $S^{\textit{post}}_t$ are strictly positive. Let us fix $t$. By the Markov property of $D$ (and using that $r_t=r(D_t)$ and $\lambda_t=\lambda(D_t)$), we may without loss of generality assume that $t=0$ (and think of $T$ as $T-t$), as long as the starting point $D_0>0$ is allowed to be arbitrary.

By conditioning on $D_T$, we find
$$
\text{Cov}\left(e^{-\int_0^T\lambda_u du} \phi(D_T), D_T U'(D_T+\varepsilon) \right)
= \text{Cov}\Big(g(D_T) \phi(D_T), D_T U'(D_T+\varepsilon) \Big)
$$
and
$$
\text{Cov}\left(e^{-\int_0^T\lambda_u du} \phi(D_T), e^{\int_0^Tr_udu} U'(D_T+\varepsilon) \right)
= \text{Cov}\left(g(D_T) \phi(D_T), f(D_T) U'(D_T+\varepsilon) \right),
$$
where $f(x)=\Ex[ e^{\int_0^Tr_udu} \mid D_0=x]$, and $g(x)=g(0,x)$ is given in~(\ref{eq:gdef}). Since $r$ is counter-cyclical, $f$ is decreasing by Lemma~\ref{L:stdom2}, and hence $x\mapsto f(x)U'(x+\varepsilon)$ is also decreasing. Moreover, the function $\psi(x)=xU'(x+\varepsilon)$ has a derivative $\psi'(x) = U'(x+\varepsilon) + xU''(x+\varepsilon)$, which is strictly greater than zero if and only if
$$
1 > -\frac{xU''(x+\varepsilon)}{U'(x+\varepsilon)} = -\frac{x}{x+\varepsilon}\frac{(x+\varepsilon)U''(x+\varepsilon)}{U'(x+\varepsilon)} = \frac{x}{x+\varepsilon} L_U(x+\varepsilon).
$$
This is indeed the case since the relative risk aversion is less than or equal to one. Thus $\psi$ is strictly increasing.

Under the assumption of $(i)$, $g(x)\phi(x)$ is strictly increasing, so the first covariance is strictly positive, while the second is strictly negative. This uses the fact that for positive, strictly increasing functions $h_1$ and $h_2$, and any non-constant random variable $X$, $\text{Cov}(h_1(X),h_2(X))>0$, while if $h_2$ is strictly decreasing, $\text{Cov}(h_1(X),h_2(X))<0$.

Under the assumption of $(ii)$ that $g(x)\phi(x)$ is strictly decreasing, the situation reverses and the jump becomes strictly negative.
\end{proof}

We also provide the following result, which shows that the stock price jump will be negative under more general conditions than those of Theorem~\ref{T:sj}.

\begin{theorem} \label{T:sj2}
Assume that the interest rate is counter-cyclical and the default intensity pro-cyclical. If the representative agent has strictly decreasing absolute risk aversion, then the stock price has a strictly negative jump at $\tau$.
\end{theorem}

\begin{proof}
Let $a_t$, $b_t$, $c_t$ and $d_t$ be as in the proof of Theorem~\ref{T:sj}. From Equation~(\ref{eq:j0}) we see that a sufficient condition for having a strictly negative jump is that $a_td_t - b_tc_t>0$ for all $t\leq T$. As in the proof of Theorem~\ref{T:sj} it suffices to consider $t=0$.

By Lemma~\ref{L:stdom4} and the cyclicality of $r$ and $\lambda$, we have
$$
\Ex\left[ e^{\int_0^T (r_u-\lambda_u) du} \mid D_T\right]\geq \Ex\left[ e^{\int_0^Tr_udu}  \mid D_T\right] \Ex\left[ e^{-\int_0^T\lambda_udu} \mid D_T\right].
$$
Therefore, with $f(x)=\Ex[ e^{\int_0^Tr_udu} \mid D_T=x]$ and $g(x)=\Ex[ e^{-\int_0^T \lambda_u du} \mid D_T=x]$, we obtain by conditioning on $D_T$ that
\begin{align*}
a_0d_0 - b_0c_0 &\geq \Ex\left[ U'(D_T+\varepsilon) D_T \right] \Ex\left[ f(D_T)g(D_T)U'(D_T+\varepsilon)\phi(D_T)\right] \\
& \qquad - \Ex\left[f(D_T) U'(D_T+\varepsilon) \right] \Ex\left[ f(D_T) U'(D_T+\varepsilon)D_T\phi(D_T)\right].
\end{align*}
Here $\phi(x)$ is again given by~(\ref{eq:gdef}). The derivative of $\phi$ is
\begin{equation} \label{eq:phider}
\phi'(x)  = \frac{U'(x+1)}{U'(x+\varepsilon)} \Big[ \ell_U(x+1) - \ell_U(x+\varepsilon) \Big],
\end{equation}
which is strictly negative since the absolute risk aversion is assumed to be strictly decreasing. Therefore $\phi$ is strictly decreasing. Moreover, by Lemma~\ref{L:stdom4} and the cyclicality of $r$ and $\lambda$, the functions $f$ and $g$ are decreasing. They are also strictly positive. Hence
\begin{align*}
\Big( xf(y)g(y)\phi(y) & - f(x)f(y)y\phi(y)\Big)  + \Big( yf(x)g(x)\phi(x) - f(y)f(x)x\phi(x) \Big) \\
&= f(x)f(y) \left( \frac{x}{f(x)} - \frac{y}{f(x)} \right) \Big( g(y)\phi(y) - g(x)\phi(x) \Big) > 0
\end{align*}
for $x\neq y$. Observing that $D_T$ has no atoms and $U'(x+\varepsilon)>0$, Lemma~\ref{L:stdom10} then yields that $a_0d_0 - b_0c_0>0$, as desired.
\end{proof}

Naively one might expect the jump in the stock price always to be negative, for the following reason. The default event leads to an instantaneous drop in the aggregate wealth in the economy. If the representative investor has a decreasing absolute risk aversion, this should lead to a reduced demand for the risky asset (after default, the stock is the only risky asset). This in turn forces the stock price down so that market clearing is maintained.

Such an argument supposes that the stock price jump is exclusively a wealth effect. However, when the default intensity is stochastic, there is also a ``non-myopic'' effect originating from expected future co-movements of the default intensity and the dividend process. Specifically, if the default intensity is highly counter-cyclical, and the current (pre-default) value of the dividend process is low, then even a moderate expected future dividend increase is coupled with a dramatic future reduction in the default intensity. The representative investor, in anticipation of the reduced risk of default, may then wish to shift wealth to the defaultable bond. This causes a downward pressure on the stock price, pushing it below what would be its fundamental value, were there no defaultable bond in the economy. When the default occurs, this downward pressure vanishes, and the stock price jumps up.

Of course, the same reasoning could be used for very high values of the dividend process to argue that the jump would be negative in these cases. Consistent with this observation, we have found that the function $x\mapsto \phi(x) g(t,x)$ appearing in Theorem~\ref{T:sj} becomes decreasing for large values of $x$, even in examples where $\lambda$ is highly counter-cyclical. In such cases the price jump will still be (mostly) positive on simulated paths, if the probability is sufficiently small that $D_t$ ever reaches the high levels where the function is decreasing.

We end this section with a numerical case study to support the argument just made. Specifically, we assume that the dividend process in Equation~(\ref{eq:divid}) is a geometric Brownian motion, i.e. $\mu^D(x) = \mu$ and $\sigma^D(x) = \sigma$. Using time reversal of diffusions, see Lemma~\ref{lem:exsj} in Appendix~\ref{AP:stdom}, we may write
$$
g(x) = \Ex\left[ e^{-\int_0^T \lambda(\widetilde{D}_u) du} \bigg| \widetilde{D}_0 = x \right],
$$
where the process $\widetilde D$ satisfies the SDE
\begin{equation} \label{eq:expgphires}
d\widetilde D_t = \widetilde{\mu}(t,\widetilde{D}_t) dt + \sigma \widetilde D_t dW_t
\end{equation}
with
\begin{equation} \label{eq:expgphires2}
\widetilde{\mu}(t,x) = -\mu x + \frac{1}{\sigma} \left(\mu - \frac{1}{2} \sigma^2 - \frac{\log(x/D_0)}{T-t}\right).
\end{equation}
We set $\mu = -0.2$, $\sigma = 0.3$, $D_0=1$, $r=0.03$, and use a strongly counter-cyclical default intensity given by $\lambda(x) = 9 e^{-x}$. Further, we choose a logarithmic utility function given by $U(x) = \log(x)$. Under these choices of parameters, we estimated via Monte-Carlo simulation that at the default time the stock experiences a positive jump of size 0.001.

We estimate $g(x)$ via Monte-Carlo simulations using~(\ref{eq:expgphires}) and (\ref{eq:expgphires2}), and report the behavior of $\phi(x) g(x)$ in Figure~\ref{fig:compphig}. We see that this function is initially increasing, and it only starts decreasing for sufficiently large values of $x$ ($x > 9$). However, the probability that the geometric Brownian motion with negative drift reaches those values before time $T$, given that it starts at $1$, is extremely low.

\begin{figure}
\begin{center}
\includegraphics[scale=0.5]{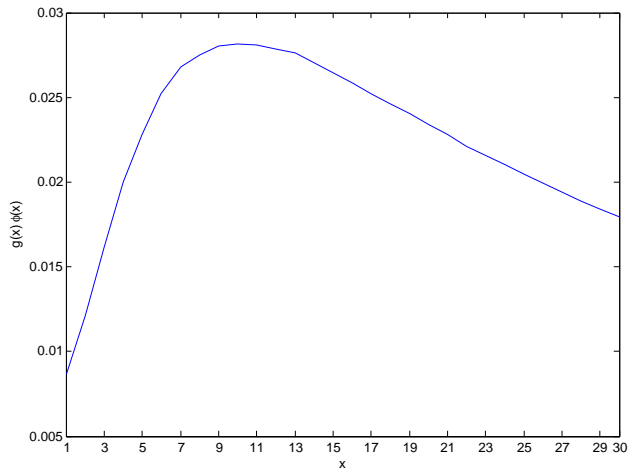}
\vspace{-4.2cm}
\caption{Plot of the function $\phi(x) g(x)$. }
\label{fig:compphig}
\end{center}
\end{figure}

\section{Wealth processes} \label{S:W}
The jump in an individual agent's wealth can be analyzed using the same techniques as for the stock price. Starting from Equations~(\ref{eq:Wkt2}) and~(\ref{eq:Wkexp}), and using Lemma~\ref{L:ce} in Appendix~\ref{AP:fe}, the wealth of the $k$:th investor can be decomposed into a pre- and post-default term. The result is
$$
W_{kt} = \1{\tau>t} W_{kt}^{\textit{pre}} + \1{\tau\leq t} W_{kt}^{\textit{post}},
$$
where
\begin{align*}
W_{kt}^{\textit{post}} &= \frac{1}{\xi^{\textit{post}}_t}\Ex[ U'(D_T+\varepsilon) I_k(y_k U'(D_T+\varepsilon)) \mid \mathcal F_t ] \\
W_{kt}^{\textit{pre}} &= \frac{1}{\xi^{\textit{pre}}_t} \Ex\bigg[\left(1-e^{-\int_t^T \lambda_s ds} \right) U'(D_T + \epsilon)I_k(y_k U'(D_T+\varepsilon)) \\
&\qquad\qquad\qquad +  e^{-\int_t^T \lambda_s ds}  U'(D_T + 1) I_k(y_k U'(D_T+1))  \bigg| \mathcal{F}_t  \bigg].
\end{align*}
The jump in wealth is then $\Delta W_{k\tau} = W_{k\tau}^{\textit{post}} -W_{k\tau}^{\textit{pre}}$ on $\{\tau\leq T\}$. The following result shows that the condition of Theorem~\ref{T:sj2} is also sufficient to ensure a negative jump in wealth. Unfortunately, the structure of the final value of the wealth process prevents us from obtaining a simple condition to guarantee a positive jump. (The reason is that, in contrast to the stock, $W_{kT}$ cannot be expressed as $\xi_T$ times an $\mathcal F_T$-measurable random variable.)

\begin{theorem} \label{T:wj}
Assume that the interest rate is counter-cyclical and the default intensity pro-cyclical. If the representative agent has strictly decreasing absolute risk aversion, then every agent's wealth process has a strictly negative jump at $\tau$.
\end{theorem}

\begin{proof}
We consider the $k$:th investor, so let us fix $k$. The proof follows along the same lines as that of Theorem~\ref{T:sj2}. The jump in wealth is
\begin{equation} \label{eq:wj0}
\Delta W_{k\tau} = \frac{a^k_t}{c_t} - \frac{a^k_t-b^k_t}{c_t-d_t} \Bigg|_{t=\tau} \qquad \text{on}\quad \{\tau\leq T\},
\end{equation}
where
\begin{align*}
a^k_t &= \Ex\left[ U'(D_T+\varepsilon) \iota_k(D_T+\varepsilon) \mid \mathcal F_t \right] \\
b^k_t &= \Ex\left[ e^{-\int_t^T \lambda_u du} ( U'(D_T+\varepsilon) \iota_k(D_T+\varepsilon) - U'(D_T+1)  \iota_k(D_T+1))  \mid \mathcal F_t \right] \\
c_t &= \Ex\left[ e^{\int_t^T r_u du} U'(D_T+\varepsilon) \mid \mathcal F_t \right] \\
d_t &= \Ex\left[ e^{\int_t^T (r_u-\lambda_u) du} ( U'(D_T+\varepsilon) - U'(D_T+1) )  \mid \mathcal F_t \right],
\end{align*}
and $\iota_k(x)=I_k(y_k U'(x))$. As in the proof of Theorem~\ref{T:sj2}, it suffices to prove that $a^k_0d_0 - b^k_0c_0>0$. To make the notation less cluttered we write $D=D_T$, $R=\int_0^Tr_udu$, $\Lambda=\int_0^T\lambda_udu$. As before, $\phi(x) = 1 - \frac{U'(x+1)}{U'(x+\varepsilon)}$. Since both $I_k$ and $U'$ are decreasing, $\iota_k$ is increasing. Hence
\begin{align*}
a^k_0d_0 - b^k_0c_0 &\geq \Ex\left[ U'(D+\varepsilon)\iota_k(D+\varepsilon)\right] \Ex\left[ e^{R-\Lambda}U'(D+\varepsilon)\phi(D)\right] \\
& \qquad - \Ex\left[e^{-\Lambda} U'(D+\varepsilon) \iota_k(D+\varepsilon) \phi(D)\right] \Ex\left[ e^R U'(D+\varepsilon)\right].
\end{align*}
The cyclicality of $r$ and $\lambda$ implies, via Lemma~\ref{L:stdom4}, that
$$
\Ex[ e^{R-\Lambda} \mid D]\geq \Ex[ e^R \mid D] \Ex[ e^{-\Lambda} \mid D].
$$
Therefore, with $f(x)=\Ex[ e^R \mid D=x]$ and $g(x)=\Ex[ e^{-\Lambda} \mid D=x]$, we obtain by conditioning on $D$ that
\begin{align*}
a^k_0d_0 - b^k_0c_0 &\geq \Ex\left[ U'(D+\varepsilon)\iota_k(D+\varepsilon)\right] \Ex\left[ f(D)g(D)U'(D+\varepsilon)\phi(D)\right] \\
& \qquad - \Ex\left[g(D) U'(D+\varepsilon) \iota_k(D+\varepsilon) \phi(D)\right] \Ex\left[ f(D) U'(D+\varepsilon)\right].
\end{align*}
Now, the cyclicality of $r$ and $\lambda$ together with Lemma~\ref{L:stdom2} shows that $f$ and $g$ are decreasing. Since also $f$ is strictly positive, $\phi$ is strictly decreasing, and $\iota_k$ is increasing, we have
\begin{align*}
\Big(\iota_k(x+\varepsilon)&f(y)g(y)\phi(y) - g(x)\iota_k(x+\varepsilon)\phi(x)f(y)\Big) \\
&\qquad + \Big( \iota_k(y+\varepsilon)f(x)g(x)\phi(x) - g(y)\iota_k(y+\varepsilon)\phi(y)f(x)\Big) \\
&= f(x)f(y)\Big( \phi(y)g(y) - \phi(x)g(x) \Big) \Big( \frac{1}{f(x)}\iota_k(x+\varepsilon)  - \frac{1}{f(y)}\iota_k(y+\varepsilon) \Big) > 0
\end{align*}
for $x\neq y$. The positivity of $a_0d_0 - b_0c_0$ now follows by Lemma~\ref{L:stdom10}, since $U'(x+\varepsilon)>0$ and $D$ has no atoms.
\end{proof}

\subsection{Jump sizes under power utility}\label{sec:jumpow}
We now investigate how the size of the jump is affected by the risk aversion of the agents. For this, we assume that all agents in the economy have power utility with relative risk aversion $\gamma_k\in (0,1]$. That is,
$$
U_k(x) = \frac{x^{1-\gamma_k}}{1-\gamma_k},
$$
which should be interpreted as $U_k(x)=\log(x)$ when $\gamma_k=1$. We then have
$$
U_k'(x) = x^{-\gamma_k} \qquad\text{and}\qquad I_k(y) = y^{-1/\gamma_k}.
$$
The following result gives a condition under which a more risk averse investor will suffer a smaller jump in wealth than one who is less risk averse.

\begin{proposition} \label{P:wjs}
Assume that the interest rate is counter-cyclical and the default intensity pro-cyclical, and that the representative agent has strictly decreasing absolute risk aversion. Consider two agents $k$ and $\ell$ with $\gamma_k\geq\gamma_\ell$. If
\begin{equation}\label{eq:Pwis}
\frac{\Ex\left[ e^{-\int_t^T\lambda_udu}U'(D_T+1)^{1-1/\gamma_k}\mid \mathcal F_t \right]}{\Ex\left[ U'(D_T+\varepsilon)^{1-1/\gamma_k}\mid \mathcal F_t \right]}
\leq \frac{\Ex\left[ e^{-\int_t^T\lambda_udu}U'(D_T+1)^{1-1/\gamma_\ell}\mid \mathcal F_t \right]}{\Ex\left[ U'(D_T+\varepsilon)^{1-1/\gamma_\ell}\mid \mathcal F_t \right]}
\end{equation}
for all $0\leq t\leq T$, then
$$
\left| \frac{\Delta W_{k\tau}}{W_{k\tau-}} \right| \leq \left| \frac{\Delta W_{\ell\tau}}{W_{\ell\tau-}} \right|.
$$
If $\lambda$ is constant, the statement remains true also in the case where both inequalities are reversed.
\end{proposition}

\begin{proof}
Let $a^k_t$, $b^k_t$, $c_t$, $d_t$ be as in the proof of Theorem~\ref{T:wj}. If $W_k$ jumps at $t$, we have
$$
\frac{\Delta W_{kt}}{W_{kt-}} = \frac{c_t-d_t}{a^k_t-b^k_t}\left(\frac{a^k_t}{c_t} - \frac{a^k_t - b^k_t}{c_t - d_t}\right) = \frac{c_t-d_t}{c_t}\frac{a^k_t}{a^k_{t}-b^k_t} - 1,
$$
and this is negative by Theorem~\ref{T:wj} (this is the only place where the counter-cyclicality of $r$ is needed.) Hence
$$
\left| \frac{\Delta W_{k\tau}}{W_{k\tau-}} \right| - \left| \frac{\Delta W_{\ell\tau}}{W_{\ell\tau-}} \right| =
\frac{c_t-d_t}{c_t}\left(\frac{a^\ell_t}{a^\ell_t-b^\ell_t}  - \frac{a^k_t}{a^k_t-b^k_t}\right),
$$
and this is nonpositive if and only if $a^\ell_t b^k_t \geq a^k_tb^\ell_t$. As in the proof of Theorem~\ref{T:sj} it is enough to consider $t=0$. Let us define $\nu_k=1-1/\gamma_k$ and $\nu_\ell=1-1/\gamma_\ell$. The assumption of power utility implies that
$$
U'(x)I_k(y_k U'(x)) = y_k^{-1/\gamma_k} U'(x)^{\nu_k},
$$
and hence, with $D=D_T$ and $\Lambda = \int_0^T\lambda_sds$,
\begin{align*}
a^\ell_0 b^k_0 - a^k_0b^\ell_0 &= y_k^{-1/\gamma_k}y_\ell^{-1/\gamma_\ell}
\Big( \Ex\left[ U'(D+\varepsilon)^{\nu_\ell}\right] \Ex\left[ e^{-\Lambda} U'(D+\varepsilon)^{\nu_k}\right]
\\
&\qquad\qquad\qquad\qquad - \Ex\left[e^{-\Lambda}  U'(D+\varepsilon)^{\nu_\ell}\right] \Ex\left[ U'(D+\varepsilon)^{\nu_k}\right] \\
&\qquad\qquad + \Ex\left[ U'(D+\varepsilon)^{\nu_k}\right] \Ex\left[ e^{-\Lambda} U'(D+1)^{\nu_\ell}\right]
\\
&\qquad\qquad\qquad\qquad - \Ex\left[e^{-\Lambda}  U'(D+1)^{\nu_k}\right] \Ex\left[ U'(D+\varepsilon)^{\nu_\ell}\right]
  \Big).
\end{align*}
The result follows once we prove that the first difference inside the parentheses is nonnegative, i.e.,
\begin{equation}\label{eq:wj10}
\Ex\left[ U'(D+\varepsilon)^{\nu_\ell}\right] \Ex\left[ e^{-\Lambda} U'(D+\varepsilon)^{\nu_k}\right]
\geq
\Ex\left[e^{-\Lambda}  U'(D+\varepsilon)^{\nu_\ell}\right] \Ex\left[ U'(D+\varepsilon)^{\nu_k}\right],
\end{equation}
where by conditioning on $D$ we may replace $e^{-\Lambda}$ by $g(D)=\Ex[ e^{-\Lambda} \mid D]$.
Since $\gamma_k\geq \gamma_\ell$, we have $\delta=\nu_k-\nu_\ell\geq 0$. Moreover, since $U'$ and $g$ are both decreasing (the latter due to Lemma~\ref{L:stdom2} and the pro-cyclicality of $\lambda$), we have that
$$
U'(x+\varepsilon)^{\nu_\ell}U'(y+\varepsilon)^{\nu_\ell}\Big(U'(y+\varepsilon)^\delta - U'(x+\varepsilon)^\delta \Big)\Big( g(y) - g(x) \Big) \geq 0.
$$
Thus, it is enough to apply Lemma~\ref{L:stdom4} to establish~(\ref{eq:wj10}), which completes the proof for non-constant $\lambda$. The last assertion is readily deduced upon noting that equality holds in~(\ref{eq:wj10}) if $\lambda$ is constant.
\end{proof}

As a corollary we obtain that in an economy populated exclusively by investors with power utilities and logarithmic utilities, those with logarithmic utilities will suffer the smallest relative jump in wealth.

\begin{corollary}
Assume that the interest rate is counter-cyclical and the default intensity constant, and that the representative agent has strictly decreasing absolute risk aversion. Consider two agents $k$ and $\ell$. If $\gamma_k=1$, i.e.~the $k$:th investor has log-utility, then
$$
\left| \frac{\Delta W_{k\tau}}{W_{k\tau-}} \right| \leq \left| \frac{\Delta W_{\ell\tau}}{W_{\ell\tau-}} \right|.
$$
\end{corollary}

\begin{proof}
If $\lambda$ is constant and $\gamma_k=1$, the inequality~(\ref{eq:Pwis}) reduces to
$$
\Ex\left[ U'(D_T+1)^{1-1/\gamma_\ell}\mid \mathcal F_t \right]
\geq
\Ex\left[ U'(D_T+\varepsilon)^{1-1/\gamma_\ell}\mid \mathcal F_t \right].
$$
This is satisfied since $U'(x)^{1-1/\gamma_\ell}$ is increasing in $x$, so Proposition~\ref{P:wjs} applies.
\end{proof}

\subsection{Measures of Systemic Risk}
Based on the analysis done in the previous sections, we suggest two measures to quantify the amount of systemic risk at time $t$ in our economy. These are given by
\begin{eqnarray*}
\nonumber \varrho^{W}_t &=& \frac{\frac{1}{N} \sum_{k=1}^N ( W_{kt}^{pre} - W_{kt}^{post} )}{P_t - \varepsilon}, \\
\nonumber \varrho^{S}_t &=&  \frac{S_t^{pre} - S_t^{post}}{P_t - \varepsilon}.
\end{eqnarray*}
Here $S_t^{pre} - S_t^{post}$ corresponds to the drop in the stock price, and $W_{kt}^{pre} - W_{kt}^{post}$ to the drop in the wealth of the $k$:th investor, if default were to happen at time $t$. Note that the measures are positive if the drop is positive (the jump is downward). The measure $\varrho^{S}_t$ measures the impact a default would have on the stock, under the scenario that a default is imminent. The measure, $\varrho^{W}_t$, instead, quantifies the impact that default would have on the aggregate wealth of the economy, under the same scenario. Both measures can be interpreted as the number of dollars lost by the stock (respectively by the portfolio of the ``average'' investor in the economy) for each dollar lost by the corporate bond at time $t$, in case default occurs at $t$.
Notice that the two measures convey different information. While $\varrho^{S}_t$ depends on the interplay between cyclicality properties of the default intensity and interest rate, and the risk aversion of the representative investor, $\varrho^{W}_t$ also accounts for the aggregate level of risk aversion in the economy. We postpone the characterization of the dependence of these measures on the market and default risk parameters of our model for future research.


%

\section{Conclusions} \label{sec:concl}
We have developed a novel framework where a stock and a defaultable bond interact endogenously through equilibrium mechanisms. Our market consists of a money market account, a stock, and a defaultable bond, which are related to each other only through an underlying dividend process, whose dynamics is unaffected by the default event. The price processes of the stock and of the defaultable bond are determined endogenously in equilibrium. We analyzed in detail the impact of the default event on the stock price, market price of risk, default risk premium, and investor wealth processes, as well as the relations between them. We found that the equilibrium price of the stock typically jumps at default. As the default event has no casual impact on the dividend process, this results in a form of endogenous interaction between the stock and the defaultable bond. We have characterized the direction of the jump of the stock price at default in terms of investor preferences and cyclicality properties of the default intensity, showing that upwards jumps are possible when the default intensity is sufficiently counter-cyclical. Under the assumption of pro-cyclical default intensity and counter-cyclical interest rate, we have shown that the wealth process of the representative investor jumps down upon default, and that power utility investors will suffer a smaller relative jump in wealth if they are more risk averse. Based on the analysis done in the paper, we have suggested two possible measures to quantify systemic risk.
In the future, we would like to extend our results to an economy consisting of multiple defaultable securities, and analyze how default correlations and cyclicality properties of the model parameters impact the price of the securities and the aggregate wealth in the economy.

\paragraph{Acknowledgments} A significant portion of the research reported in this paper was done while the authors were visiting the Swiss Institute of Finance at EPFL. The authors are grateful to them for the hospitality and the useful conversations on this topic.
We would also like to thank Jak\v sa Cvitani\'c, Jeremy Staum, and Robert Jarrow for useful conversations and feedback provided.

\appendix

\section{Results relating to filtration expansion} \label{AP:fe}

\begin{theorem}[Martingale representation in $\mathbb G$] \label{T:mgr}
For every square integrable $\mathbb G$~martingale $N$ there are $\mathbb G$~predictable processes $(a_t)_{0\leq t\leq T}$ and $(b_t)_{0\leq t\leq T}$, such that
$$
\Ex\Big(\int_0^T |a_s|^2 ds\Big) < \infty, \qquad \Ex\Big(\int_0^T |b_s|^2 \lambda_s ds\Big) < \infty,
$$
and
$$
N_t = N_0 + \int_0^t a_s dB_s + \int_0^t b_s dM_s.
$$
\end{theorem}

\begin{proof}
This follows from Theorem~2.3 in~\cite{Kusuoka:1999}, since every $\mathbb F$~martingale remains a $\mathbb G$~martingale.
\end{proof}

The following result is crucial in that it allows us to reduce $\mathcal G_t$-conditional expectations to $\mathcal F_t$-conditional expectations. This type of result is classical in credit risk modeling.

\begin{lemma} \label{L:ce}
Let $X=X^1\1{\tau>T} + X^2\1{\tau\leq T}$, where $X^1$ and $X^2$ are integrable $\mathcal F_T$-measurable random variables. Then
$$
E\big[X \mid\mathcal G_t\big] = \1{\tau\leq t}\Ex\left[X^2 \mid \mathcal F_t\right] + \1{\tau> t}\Ex\left[\big(1-e^{-\int_t^T \lambda_s ds} \big) X^2
 + e^{-\int_t^T \lambda_s ds} X^1 \mid\mathcal F_t \right]
$$
\end{lemma}

\begin{proof}
First note that
$$
X = X^2 - \1{\tau>T} (X^2-X^1).
$$
Since Hypothesis (H) holds between $\mathbb F$ and $\mathbb G$, any $\mathbb F$~martingale $N$ satisfies $\Ex[N_T\mid\mathcal G_t]=N_t=\Ex[N_T\mid\mathcal F_t]$. Apply this with $N_t=\Ex[ X^2 \mid \mathcal F_t ]$, whose final value is $N_T= X^2$ since $X^2$ is $\mathcal F_T$-measurable, to get $\Ex[X^2\mid\mathcal G_t]=\Ex[X^2\mid\mathcal F_t]$. Next, use the identity
$$
\Ex[ \1{\tau>T} Y\mid\mathcal G_t]= \1{\tau>t}\Ex[e^{-\int_t^T\lambda_u du} Y\mid\mathcal F_t],
$$
which holds for $\mathcal F_T$-measurable and integrable $Y$, with $Y=X^2-X^1$. The claim now follows after some rearrangement.
\end{proof}

\begin{lemma} \label{L:jp}
Let $X_t=X^{\textit{pre}}_t\1{\tau>t}+X^{\textit{post}}_t\1{\tau\leq t}$ be a $\mathbb G$~semimartingale, where $X^{\textit{pre}}$ and $X^{\textit{post}}$ are continuous. Then
\begin{equation}\label{eq:L1}
dX_t =  \1{\tau\geq t}dX^{\textit{pre}}_t+\1{\tau < t}dX^{\textit{post}}_t
+ (X^{\textit{post}}_t - X^{\textit{pre}}_t)d\1{\tau\leq t}.
\end{equation}
If $X$ is a strictly positive with representation
$$
\frac{dX_t}{X_{t-}} = a_t dB_t + b_t dM_t + c_t dt,
$$
then $b_t = (X^{\textit{post}}_t / X^{\textit{pre}}_t) - 1$. If in addition $X^{\textit{pre}}$ and $X^{\textit{post}}$ have representations
$$
\frac{dX^i_t}{X^i_t} = a^i_t dB_t + c^i_t dt, \qquad i\in\{\textit{pre}, \textit{post}\},
$$
then $a_t = a^{\textit{pre}}_t\1{\tau\geq t}+a^{\textit{post}}_t\1{\tau< t}$.
\end{lemma}

\begin{proof}
The expression~(\ref{eq:L1}) follows from It\^o's formula. Concerning the expression for $b_t$, note that the continuity of $X^{\textit{pre}}$ and $X^{\textit{post}}$ implies that
$$
X_t = Y_t + \int_0^t (X^{\textit{post}}_t - X^{\textit{pre}}_t)dM_t
$$
for some continuous process $Y$. Since $\int_0^t (X^{\textit{post}}_s - X^{\textit{pre}}_s) dM_s = \int_0^t X_{s-} ((X^{\textit{post}}_s / X^{\textit{pre}}_s) - 1) dM_s$, the result follows. Finally, combining~(\ref{eq:L1}) with the assumed representation for $X^{\textit{pre}}$ and $X^{\textit{post}}$ yields
$$
dX_t = \1{\tau\geq t}X^{\textit{pre}}_t a^{\textit{pre}}_t dB_t + \1{\tau < t}X^{\textit{post}}_t a^{\textit{post}}_t dB_t + d \widetilde Y_t,
$$
where $\widetilde Y_t$ is a stochastic integral with respect to $dM_t$ and $dt$ only. The expression for $a_t$ now follows, since $\1{\tau\geq t}X^{\textit{pre}}_t=\1{\tau\geq t}X_{t-}$ and $\1{\tau< t}X^{\textit{post}}_t=\1{\tau< t}X_{t-}$.
\end{proof}

\section{Results relating to stochastic ordering and correlations} \label{AP:stdom}

\begin{lemma} \label{L:stdom2}
Let $X=(X_t)_{0\leq t\leq T}$ satisfy $dX_t = a(t,X_t)dB_t + b(t,X_t)dt$ with a fixed starting point $X_0$, where we assume that
\begin{itemize}
\item[(i)] $a(t,x)$ and $b(t,x)$ are infinitely differentiable;
\item[(ii)] $X$ is does not explode;
\item[(iii)] for each $t>0$, $X_t$ admits a density $p(t,x)$ with continuous second derivatives.
\end{itemize}
Let $\Phi$ be an nondecreasing (nonincreasing) function of $x$. Then
$$
f(x)=\Ex\left[e^{\int_0^T \Phi(X_s)ds} \bigg| X_T=x \right]
$$
is nondecreasing (nonincreasing) in $x$.
\end{lemma}

\begin{proof}
The proof is based on time reversal of diffusions. Let $Y_t = X_{T-t}$. Then $X_T=Y_0$ and $\int_0^T \Phi(X_s)ds=\int_0^T \Phi(Y_s)ds$, so
$$
f(x)=\Ex[e^{\int_0^T \Phi(Y_s)ds} \mid Y_0 =x].
$$
We wish to apply Theorem~2.1 in~\cite{Haussmann/Pardoux:1986} to obtain the dynamics of the time-reversed process $Y$. The smoothness of $p(t,x)$ and the local Lipschitz property of $a$ and $b$ (which is guaranteed by their smoothness), together with condition~$(ii)$, imply that the assumptions of that theorem are satisfied; see~\cite{Haussmann/Pardoux:1986}, Remark~2.2 and Section~3. This yields
\begin{equation} \label{eq:Lstdom21}
dY_t = \widetilde a(t,Y_t) d\widetilde B_t + \widetilde b(t,Y_t) dt,
\end{equation}
where
\begin{eqnarray}
\nonumber \widetilde b(t,x) &=& -b(T-t,x) + \frac{[a(T-t,x) p(T-t,x)]_x}{p(T-t,x)} \\
\widetilde a(t,x) &=& a(T-t,x),
\label{eq:bacoeff}
\end{eqnarray}
and $\widetilde B$ is Brownian motion. The smoothness of $p$, $a$ and $b$ implies that $\widetilde a$ and $\widetilde b$ are continuously differentiable on the interior of the support of $X$, and hence locally Lipschitz there. By localization we may assume they are globally Lipschitz, so that standard comparison theorems (see for instance~\cite{Ikeda/Watanabe:1977}) become available. Specifically, if $x_1\leq x_2$ lie in the support of $X_T$, and $Y^i$ denotes the solution to~(\ref{eq:Lstdom21}) started from $x_i$, we have $P(Y^1_t\leq Y^2_t, \ 0\leq t<T)=1$ and hence $f(x_1)\leq f(x_2)$ if $\Phi$ is nondecreasing. The case of nonincreasing $\Phi$ is deduced in the same manner.
\end{proof}

\begin{lemma} \label{L:stdom3}
Let $X$ be as in Lemma~\ref{L:stdom2}. Suppose $F_0,\ldots, F_n$ and $G_0,\ldots,G_n$ are all nondecreasing (resp.~all nonincreasing), nonnegative functions, and let $0\leq t_0\leq \cdots \leq t_n \leq T$. Then
$$
f(x) = \Ex\left[ \prod_{i=0}^n F_i(X_{t_i}) \bigg| X_T=x\right]
$$
is nondecreasing (resp.~nonincreasing), and we have
$$
\Ex\left[ \prod_{i=0}^n F_i(X_{t_i}) \prod_{i=0}^n G_i(X_{t_i}) \bigg| X_T\right] \geq
\Ex\left[ \prod_{i=0}^n F_i(X_{t_i}) \bigg| X_T\right]
\Ex\left[ \prod_{i=0}^n G_i(X_{t_i}) \bigg| X_T\right].
$$
\end{lemma}

\begin{proof}
We treat the nondecreasing case, the other one being similar. Consider again the time-reversed process $Y_t = X_{T-t}$, and define time points $s_i=T-t_{n-i}$, $i=1,\ldots,n$ and functions $\widetilde F_i=F_{n-i}$, $\widetilde G_i=G_{n-i}$. Then $0\leq s_0\leq\cdots\leq s_n\leq T$, and we have
$$
f(x) = \Ex\left[ \prod_{i=0}^n \widetilde F_i(Y_{s_i}) \bigg| Y_0=x\right].
$$
The nondecreasing property of $f$ can now be deduced as in the proof of Lemma~\ref{L:stdom2}. Concerning the inequality, we are done if we can prove that
$$
\Ex\left[ \prod_{i=0}^n \widetilde F_i(Y_{s_i}) \prod_{i=0}^n \widetilde G_i(Y_{s_i}) \bigg| Y_s\right] \geq
\Ex\left[ \prod_{i=0}^n \widetilde F_i(Y_{s_i}) \bigg| Y_s\right]
\Ex\left[ \prod_{i=0}^n \widetilde G_i(Y_{s_i}) \bigg| Y_s\right]
$$
for any $s\leq s_0$ (take $s=0$ to recover the desired inequality.) This is achieved by induction similarly as in the proof of Lemma~A.4 in~\cite{Cvitanic/Malamud:2010}. Suppose the inequality holds for $n-1$, $n-2$, etc. Then by the Markov property of $Y$ and the induction hypothesis,
\begin{align*}
\Ex\left[ \prod_{i=0}^n \widetilde F_i(Y_{s_i}) \prod_{i=0}^n \widetilde G_i(Y_{s_i}) \bigg| Y_s\right]
&= \Ex\left[ \widetilde F_0(Y_{s_0}) \widetilde G_0(Y_{s_0}) \Ex\left[ \prod_{i=1}^n \widetilde F_i(Y_{s_i}) \prod_{i=1}^n \widetilde G_i(Y_{s_i}) \bigg| Y_{s_0} \right] \bigg| Y_s\right] \\
&\geq \Ex\left[ \widetilde F_0(Y_{s_0}) F(Y_{s_0}) \widetilde G_0(Y_{s_0})  G(Y_{s_0}) \bigg| Y_s\right],
\end{align*}
where $F(x)= E\big[ \prod_{i=1}^n \widetilde F_i(Y_{s_i}) \mid Y_{s_0} = x \big]$ and $G(x)= E\big[ \prod_{i=1}^n \widetilde G_i(Y_{s_i}) \mid Y_{s_0} = x \big]$. These functions are nondecreasing by the first part of the lemma, so an application of the induction hypothesis with $n=0$ yields
\begin{align*}
\Ex\left[ \prod_{i=0}^n \widetilde F_i(Y_{s_i}) \prod_{i=0}^n \widetilde G_i(Y_{s_i}) \bigg| Y_s\right]
&\geq \Ex\left[ \widetilde F_0(Y_{s_0}) F(Y_{s_0}) \bigg| Y_s\right] \Ex\left[ \widetilde  G_0(Y_{s_0})  G(Y_{s_0}) \bigg| Y_s\right] \\
&= \Ex\left[ \prod_{i=0}^n \widetilde F_i(Y_{s_i}) \bigg| Y_s\right] \Ex\left[ \prod_{i=0}^n \widetilde G_i(Y_{s_i}) \bigg| Y_s\right],
\end{align*}
as desired. It remains to establish the case $n=0$; but this follows immediately from Lemma~A.3 in~\cite{Cvitanic/Malamud:2010}.
\end{proof}

\begin{lemma} \label{L:stdom4}
Let $X$ be as in Lemma~\ref{L:stdom2}. Let $\Phi$ and $\Psi$ be nondecreasing (resp. nonincreasing) functions. Then
$$
\Ex\left[e^{\int_0^T \Phi(X_s)ds + \int_0^T \Psi(X_s)ds}  \bigg| X_T \right] \geq
\Ex\left[e^{\int_0^T \Phi(X_s)ds} \bigg| X_T \right] \Ex\left[ e^{\int_0^T \Psi(X_s)ds}  \bigg| X_T \right].
$$
\end{lemma}

\begin{proof}
Approximate $e^{\int_0^T \Phi(X_s)ds}$ and $e^{\int_0^T \Psi(X_s)ds}$ from below using functions of the form $F_0(X_{t_0})\cdots F_n(X_{t_n})$, then apply Lemma~\ref{L:stdom3} and monotone convergence.
\end{proof}

\begin{lemma} \label{L:stdom10}
Let $f$, $g$, $G$, $h$ and $H$ be measurable functions and define
$$
\psi(x,y) = g(x)G(y) - h(x)H(y).
$$
If $f(x)$ and $\psi(x,y)+\psi(y,x)$ are nonnegative for all $x$ and $y$, then
\begin{equation}\label{eq:stdom10}
\Ex\Big[f(X)g(X)\Big] \Ex\Big[f(X)G(X)\Big] - \Ex\Big[f(X)h(X)\Big]\Ex\Big[f(X)H(X)\Big] \geq 0
\end{equation}
for every random variable $X$ for which the left side is well-defined. If $f(x)>0$ and $\psi(x,y)+\psi(y,x)>0$ for $x\neq y$, and $X$ has no atoms, the inequality is strict.
\end{lemma}

\begin{proof}
Let $\widehat X$ be an independent copy of $X$. The left side of~(\ref{eq:stdom10}) then equals
$$
\Ex\Big[f(X)g(X)f(\widehat X)G(\widehat X) - f(X)h(X)f(\widehat X)H(\widehat X)\Big] = \Ex\left[ f(X)f(\widehat X) \psi(X,\widehat X) \right],
$$
and since $X$ and $\widehat X$ are exchangeable, it is also equal to
$$
\Ex\left[ f(X)f(\widehat X) \psi(\widehat X,X) \right]
$$
Adding the two expressions yields
$$
\Ex\Big[f(X)f(\widehat X) \Big( \psi(X,\widehat X) + \psi(\widehat X,X) \Big) \Big],
$$
which is nonnegative due to the assumptions on $f$ and $\psi$. The statement concerning strict inequality is immediate.
\end{proof}

\begin{lemma} \label{lem:exsj}
Assume that the dividend process is a geometric Brownian motion, i.e. $dD_t = \mu D_t dt + \sigma D_t dW_t$, $D_0>0$. Then
\begin{equation} \label{eq:expgphi}
\Ex\left[ e^{-\int_0^T \lambda(D_u) du} \bigg| D_T = x \right] = \Ex\left[ e^{-\int_0^T \lambda(\tilde{D}_u) du} \bigg| \tilde{D}_0 = x \right],
\end{equation}
where the process $\widetilde{D}$ satisfies the SDE $d\widetilde{D}_t = \widetilde{\mu}(t,\widetilde{D}_t) dt + \sigma \widetilde{D}_t dW_t$, with
\begin{eqnarray}
\nonumber \widetilde{\mu}(t,x) &=& -\mu x + \frac{1}{\sigma} \left(\mu - \frac{1}{2} \sigma^2 - \frac{\log(x/D_0)}{T-t}\right) \\
\tilde{\sigma}(t,x) &=& \sigma x
\label{eq:tildegbm}
\end{eqnarray}
\end{lemma}

\begin{proof}
Define $\widetilde{D}_t = D_{T-t}$. Then $D_T = \widetilde{D}_0$ and $\int_0^T \lambda(\widetilde{D}_u) = \int_0^T \lambda(D_u) du$. Therefore, Eq.~(\ref{eq:expgphi}) holds. The smoothness of the transition density of $D_t$ given by
$$
p(t,y) = \frac{1}{y \sigma \sqrt{2 \pi t}} e^{- \frac{\left( \log(y/D_0) - \left(\mu - 0.5 \sigma^2 \right) t \right)^2}{2 \sigma^2 t}}
$$
along with the local Lipschitz property of $\mu x$ and $\sigma x$, and the fact that the geometric Brownian motion is nonexplosive, allow applying Theorem~2.1 in~\cite{Haussmann/Pardoux:1986}. Using Eq.~(\ref{eq:bacoeff}), we obtain the expressions in Eq.~(\ref{eq:tildegbm}).
\end{proof}

\bibliographystyle{plainnat}
\bibliography{bibl}

\end{document}